%% file: main.tex
\documentclass[12pt]{elsarticle}
\usepackage{todonotes}
\usepackage{graphicx}
\usepackage{subcaption}
\usepackage{amsmath}
\usepackage{url}
\usepackage{algorithm}
\usepackage{algpseudocode}
\usepackage{amsmath}
\usepackage{amssymb}
\usepackage{amsthm}
\usepackage{booktabs}
\usepackage{hyperref}
\usepackage{enumitem}
\usepackage{tikz}

\usepackage{array}
\usetikzlibrary{graphs, arrows.meta, shapes.geometric}

\algrenewcommand\algorithmicforall{\textbf{foreach}}
\algrenewcommand\algorithmicindent{.8em}

\newcommand{\ceil}[1]{\left\lceil #1 \right\rceil}

\newcommand{\widefig}[0]{0.8\linewidth}

\newtheorem{theorem}{Theorem}[section]
\newtheorem{corollary}{Corollary}[theorem]
\newtheorem{lemma}[theorem]{Lemma}
\newtheorem{proposition}[theorem]{Proposition}
\newtheorem{definition}{Definition}

\newcounter{example}[section]

\usepackage{amssymb}

\title{Linear Decomposition of the Majority Boolean Function using the Ones on Smaller Variables
}

\author[NTU]{Anupam Chattopadhyay}
\author[imec]{Debjyoti Bhattacharjee}
\author[isi]{Subhamoy Maitra}

\address[NTU]{CCDS, NTU, Singapore}
\address[imec]{imec, Kapeldreef 75, Leuven, 3000, Belgium}
\address[isi]{Indian Statistical Institute, India}

\begin{document}

\begin{frontmatter}

\begin{abstract}
A long-investigated problem in circuit complexity theory is to decompose an $n$-input or $n$-variable Majority Boolean function (call it $M_n$) using $k$-input ones ($M_k$), $k < n$, where the objective is to achieve the decomposition using fewest $M_k$'s. An $\mathcal{O}(n)$ decomposition for $M_n$ has been proposed recently with $k=3$. However, for an arbitrary value of $k$, no such construction exists even though there are several works reporting continual improvement of lower bounds, finally achieving an optimal lower bound $\Omega(\frac{n}{k}\log k)$ as provided by  Lecomte et. al., in CCC '22. In this direction, here we propose two decomposition procedures for $M_n$, utilizing counter trees and restricted partition functions, respectively. The construction technique based on counter tree requires $\mathcal{O}(n)$ such many $M_k$ functions, hence presenting a construction closest to the optimal lower bound, reported so far. The decomposition technique using restricted partition functions present a novel link between Majority Boolean function construction and elementary number theory. These decomposition techniques close a gap in circuit complexity studies and are also useful for leveraging emerging computing technologies.
\end{abstract}




\end{frontmatter}

\input{maj_decomposition}

\bibliographystyle{elsarticle-num}
\bibliography{sigproc}

\end{document}

%% file: maj_decomposition.tex
\section{Introduction}
\noindent Majority Boolean functions hold a special place among the classes of Boolean functions. Purely from the circuit complexity theory standpoint, Majority Boolean functions belong to the complexity class \texttt{TC$^0$}, and conjectured to strictly separate it from the complexity class \texttt{ACC}~\cite{acc_lower_bound}. There has been wide-ranging studies  related to Majority Boolean functions, starting with the early works to explore its capability to capture complex  functions~\cite{maj_depth_two_expander,maj_prefix}, implementing circuits with low count/depth of Majority and Threshold functions~\cite{constant_depth_division,constant_depth_transformer,opt_depth_threshold_mult,log_depth_threshold_div}, and linking Majority with other Boolean functions such as \textit{Threshold}, \textit{And}, \textit{Or}~\cite{threshold_by_majority,kulikov_decomp,maj_notes, majority_by_majority}. In this context, a challenging and long-studied problem is how to decompose a large, \mbox{$n$-input} Majority Boolean functions in terms of smaller functions with $k$ inputs, where $k < n$. This problem gained relevance in recent times due to the emergence of multiple computing technologies~\cite{molecular_maj,aqfp_maj,spinwave_majority,MIG_applications,testa_phd} with native realization of Majority Boolean functions. Furthermore, Majority-based logic circuit representations have demonstrated superior performance~\cite{MIG_applications,mig} compared to traditional And-Inverter Graph (AIG), prompting commercial adoption of Majority-Inverter Graph (MIG) in the synthesis toolsuite~\cite{snps-epfl}.

Circuit realization using Majority logic gates was investigated in 1960s by Akers~\cite{akers_synth} resulting in several follow-up works~\cite{akers_logic_array,riseman_maj_akers}. This line of work, after being superseded by And-Or based logic circuits, got into prominence again due to effective Majority-based logic synthesis flows~\cite{mig,mathias_xor_maj} and novel computing fabrics offering Majority logic blocks~\cite{aqfp_maj,spinwave_majority}. While these logic circuits are mostly restricted to $3$-input Majority gates, there are demonstrations with larger input sizes~\cite{maj_dram} with parallel development axiomatic system for arbitrary inputs~\cite{maj_axiom}. Therefore, decomposition of Majority Boolean functions presents an important research objective - both from theoretical and practical viewpoint.

The decomposition of $n$-input Majority ($M_n$) Boolean functions using $3$-input Majority ($M_3$) functions have been studied in~\cite{Testa:261167,VALIANT1984363} and eventually was realized with circuit size of $\mathcal{O}(n)$~\cite{decomp_tcad}. However, an efficient decomposition of $M_n$ using arbitrary $M_k$, where $k < n$ remained an open problem. Towards that goal, a lower bound has been presented recently in~\cite{complexity_ccc}. In this work, we propose two constructions, achieving circuit size closest to the lower bound reported. More specifically, while the optimal lower bound reported in~\cite{complexity_ccc} is $\Omega(\frac{n}{k}\log k)$, we demonstrate that one of our proposed constructions could achieve a compositional complexity of $\mathcal{O}(n)$, which presents the first linear decomposition of $M_n$ using arbitrary $M_k$. Our results are immediately extensible to the decomposition of Threshold Boolean functions as well. The tools used here are related to combinatorial techniques and discrete structures including the counter graph approach and partition functions.

The rest of the manuscript is organized as follows. Section~\ref{sec:prelim} presents the formal introduction of the topics necessary for our constructions. The first construction, using counter graph approach, is elaborated in Section~\ref{sec:counter_tree}. Section~\ref{sec:partition_function} is where we propose our second construction using partition functions. Both the aforementioned constructions are also studied theoretically to identify the count of $M_k$ functions. Finally, the counter graph construction is implemented using state-of-the-art logic circuit packages and benchmarked with large circuits. The results are presented in Section~\ref{sec:exp}. The work is summarized with future directions in Section~\ref{sec:conclusion}.

\section{Preliminaries}\label{sec:prelim}
\noindent We define the basic Boolean functions and terminologies used in the rest of the paper. Let $\mathbb{B} = \{0, 1\}$. An $n$-input Boolean function \mbox{$\mathbb{B}^n \rightarrow \mathbb{B}$ maps} the input truth values to a single output truth value.
Let $n$ inputs of a Boolean function $f$ be $x_1, x_2, \dots, x_n$.

\begin{definition}
For $n$ Boolean inputs $x_i~(1 \le i \le n)$, the Hamming Weight~(HW) is defined as
\begin{equation}
HW(x_1,\ldots, x_n) = \sum_{i=1}^{n}{x_i}
\end{equation}
\end{definition}

\begin{definition}
A \textit{Threshold Boolean function}~$T_n$ is defined as follows,
\begin{equation}
    T_n(x_1,\ldots,x_n) = \begin{cases}
    1 & \text{when} \left( HW(x_1, \ldots, x_n) \right) \geq t \\
    0 & \text{otherwise}
    \end{cases}
\end{equation} where $t$ is the defined threshold and $x_i$ represents the Boolean inputs~$  (1 \le i \le n)$.
\end{definition}
In the following, we denote $T_n$ with threshold $t$ as $T_n^{t}$.

\begin{definition}
An $n$-input \textit{Majority Boolean function}, denoted as $Maj_n$, where $n$ is odd, can be defined as a special case of the Threshold function, where the $t = \lceil n/2 \rceil$.
\end{definition}
\noindent For example, a $5$-input Majority function, $Maj_5$ yields true if at least $3$ of its inputs are true. $Maj_n$ is a monotone Boolean function, a property that has been utilized in earlier decomposition approaches~\cite{Testa:261167}.

In our decomposition procedure, the constituent variables of a Majority Boolean function are distinguished between, \textit{free} variables and constant \textit{control} variables, as in~\cite{complexity_ccc}. For example, in $M_7\{x_1, x_2, x_3, x_4, c_1, c_2, c_3\}$ the variables $x_i$ are input driven, which cannot be controlled. Whereas, the decomposition algorithm can fix the values of $c_i$ to be $0$ or $1$ as needed, hence termed as constant control variables. For the sake of simplicity, we do not introduce any specific definition for a Majority Boolean function, which contains some constant inputs. For example, $M_7$ with $c_1 = c_2 = c_3 = 1$, will behave as a $T_4^1$, if we consider $x_1, x_2, x_3, x_4$ as variables. On the other hand, with $c_1 = c_2 = 1, c_3 = 0$, $M_7$ will behave as $T_4^2$.
In general, we have the following technical result.

\begin{proposition}
    Consider $M_n$ with $\psi$ many control variables, and $\tau$ of them are fixed at 1 (naturally $n \geq \psi \geq \tau$). Then $M_n$ represents $T_{n-\psi}^{\lceil\frac{n}{2}\rceil-\tau}$.
\end{proposition}
\noindent To have a meaningful expression,
$n -\psi \geq \lceil\frac{n}{2}\rceil -\tau$, i.e.,
$$\psi - \tau \leq \lfloor\frac{n}{2}\rfloor$$
With this let us now move to another combinatorial object in this regard.

\begin{definition} A Majority graph is defined as a directed, acyclic graph~(DAG) $G:(V, E)$, where each vertex $V$ denotes a $Maj_n$ Boolean function with edges $E$ connecting the vertices. Each vertex has an in-degree of $n$.
\end{definition}

\noindent The majority graph can accommodate only regular edges. For majority graph that allow complemented edges, we refer to these graph as Majority-Inverter Graph~(MIG)~\cite{mig}. Note that, Majority function in this work only refers to Majority Boolean functions.

\begin{definition}
A $(n:k)$ counter takes $n$ input bits and generates a $k$-bit binary representation of the number of input bits, which are valued as $1$.
\end{definition}
\noindent A binary full adder is essentially a $(3:2)$ counter, while the half adder is $(2:2)$ counter.
\begin{definition}
A Counter graph is a DAG where each node represents a counter operation, with multiple inputs and outputs.
\end{definition}
Counter graphs~(also referred to as ``compressor trees'') have shown to be useful in Boolean arithmetic circuit optimization~\cite{verma_counter_2007}. In our designs, Hamming weight computation is necessary, which is realized using counter graphs.

\begin{definition}
In number theory, a partition $\lambda$ of a non-negative integer $n$ is defined as the sequence $\lambda = \{\lambda_1, \lambda_2, \cdots\}$ such that $\lambda_i \geq 0$; $\lambda_1 \geq \lambda_2 \geq \cdots$ and $\sum_{i\geq1}\lambda_i = n$.
\end{definition}

\begin{definition}
\textit{Partition Function}~$p(n)$ represents the number of distinct ways of representing $n$ as a sum of positive integers.
\end{definition}

\noindent For example, $3$ can be partitioned as $\{1+1+1\}$, $\{2+1\}$ and $\{3\}$. Hence, $p(3) = 3$. Note that it is conventional to write the parts within $\lambda$ in a descending order and it is also conventional to suppress $0$ values within the partitions. Partitions of a number is very well studied problem in number theory with celebrated results from Ramanujan~\cite{Ramanujan1921CongruencePO}. Of our special interest in this work is a specific type of \textit{restricted partition functions}, which is defined in the following.



\begin{definition}
Restricted partition function $p_r(N, M, n)$ is defined as the number of partitions of the number $n$ using at most $M$ parts, where each part is at most $N$.
\end{definition}

\noindent For example, $p(4) = 5$. However, $p_r(3, 2, 4) = 2$, due to the partitions $\{2+2\}$ and $\{3+1\}$.


Closed-form expressions of $p(n)$ or $p_r(N, M, n)$ are not known. Indeed, determining approximations and bounds of partition functions is one of the most challenging problems in analytic number theory~\cite{pn_hardy_ramanjuan,pn_erdos,approx_pn_summary}. For our purpose, we refer to the following approximation of $p(n)$ after~\cite{pn_hardy_ramanjuan}.

\begin{equation}\label{eqn:ramanujan}
p(n) \simeq \frac{1}{4n\sqrt{3}}e^{\pi\sqrt{\frac{2n}{3}}}, n \rightarrow \infty
\end{equation}

\begin{algorithm}[t]
    \caption{Construction of counter graph for bin position $b$}
    \label{alg:counter_graph}
    {\small
    \begin{algorithmic}[1]
    \State  $inputs_b \gets bin_{map}[b]$
    \While{true}
        \State $inputs\_remaining \gets$ $inputs_b.size()$
        \If{$inputs\_remaining = 1$}
            \State \textbf{break}
        \EndIf
        \State $counter_{in} \gets \min(inputs\_remaining, l)$
        \State $counter\_operands \gets  inputs_b$[$0:counter_{in}$]
        \State \textbf{Erase} first $counter_{in}$ elements from $inputs_b$
        \State \textbf{Create} $counterOp \gets Counter(counter\_operands)$
        \State $counter_{out} \gets counterOp.numOutput()$
        \State \textbf{Add} $counterOp.getResult(0)$ to $inputs_b$
        \State $n_b \gets b$
        \For{$i \gets 1$ to $counter_{out} - 1$}
            \State $n_b \gets n_b + 1$
            \State \textbf{Add} $counterOp.getResult(i)$ to $bin_{map}[n_b]$
        \EndFor
    \EndWhile
    \end{algorithmic}
    }
\end{algorithm}

\section{Decomposition Procedure: Counter Graph Approach}\label{sec:counter_tree}
\noindent This decomposition procedure broadly follows the bound calculation approach narrated in~\cite{complexity_ccc}.
There are two main phases of this decomposition. \textit{First}, the input bits are partitioned into several groups and for each group, the Hamming Weight~(HW) is computed.
The algorithm to compute the HW of Boolean inputs is presented in Algorithm~\ref{alg:counter_graph}, using a counter graph where each counter has at most $l$ Boolean inputs. Let us consider that for each bit position~$b$, the corresponding bits are stored in the $bin_{map}$. The algorithm proceeds in a loop, where at each step, it determines the number of inputs to process, $counter_{in}$, as the minimum between the remaining inputs and the maximum allowed number of counter input~$l$. A counter operation, $Counter$, is then created using the selected inputs, and the first $counter_{in}$ elements are removed from the input set. The outputs of the counter operation are distributed: the first output is added back to the current bin~$b$, while the remaining outputs are assigned to subsequent bins in $bin_{map}$. This process continues until only one input remains in the bin, at which point the algorithm terminates. This process is iterated till all the bins have been processed.

\textit{Second}, the HWs are added, followed by an addition with a fixed threshold value to check the final carry bit, indicating a majority. In~\cite{decomp_tcad}, the same idea has been utilized, albeit for $M_n$ to $M_3$ decomposition.

The overall idea for the specific case of $M_9$ to $M_5$ decomposition
is shown in Figure~\ref{fig:maj_9_5}.  The inputs $x_i^1$~($0\le i \le 9$) and $t^1$ are in bin~$1$ and the final output of the Algorithm is $o_1$, corresponding to bin~$1$. At the end of processing bin~$1$, new outputs $c^2_4$, $c^2_6$, $c^2_8$ and $c^2_{10}$ have been added to bin~$2$ which are then processed using the Algorithm to produce output $o_2$. The processs is repeated till all the bins have a single output. Figure~\ref{fig:counter} shows how a single counter can be decomposed using majority operations. We formally define the  decomposition of the counter to majority function in Lemma~\ref{lemma:counter-maj}.

\begin{figure}[ht]
    \begin{subfigure}[t]{.5\columnwidth}
    \caption{}
    \label{fig:maj_9_5}
    \centering
    \includegraphics[width=0.8\textwidth]{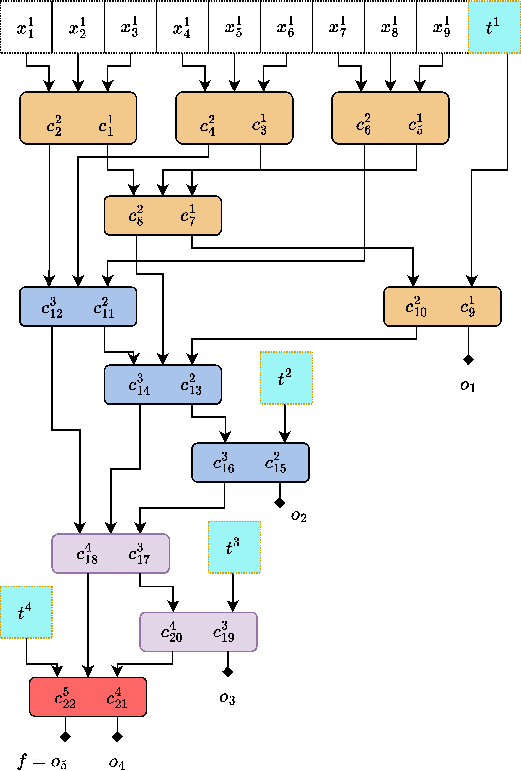}
    \end{subfigure}
    \begin{subfigure}[t]{.5\columnwidth}
        \caption{}
        \label{fig:counter}
        \centering
    \includegraphics[width=0.7\linewidth]{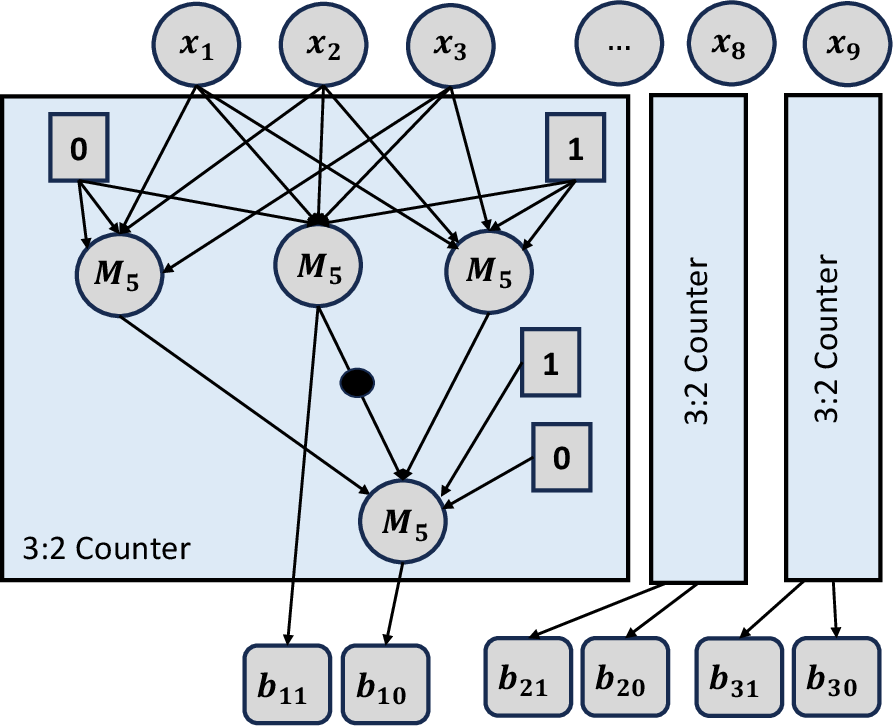}
    \end{subfigure}
    \caption{(\subref{fig:maj_9_5}) Computing $M_9$ of inputs $x_i^1$~($1 \le i \le 9$) using $(3:2)$ and $(2:2)$ counters and a threshold value~($t^4t^3t^2t^1$). The LSB of the counter output is on the right and the MSB is the leftmost output. For $M_9$, the threshold value is $1011$ since, this added with input HW generates carry bit at the last counter if $M_9$ is true. The set of counters used for computing each output bit is colored with the same color. $o_5$ is essentially the overflow bit indicating $M_9$. The counters can be expressed using $M_5$.
    (\subref{fig:counter})~Partially decomposed Counter Graph for $M_9$, where the $(3:2)$ counter is realized using $M_5$. The dot on the edge represents inversion of the input or in other words, a Boolean negation. The HW of each $3$-input partition are computed independently.}
\end{figure}


In what follows, we try to establish the compositional complexity obtained through this counter graph decomposition procedure. For the entire discussion, the input function is assumed to be $M_n$ and the target majority function as $M_k$. Before proceeding with the bound calculation, we present a few basic results.

\begin{corollary}\label{corollary:counter}
For decomposition of $M_n$ using $M_k$ via counter graphs, incoming $n$ inputs has to be partitioned in sets of at most $\lceil \frac{k}{2} \rceil$ size.
\end{corollary}
\begin{proof}
Since $M_k$ is the constituent majority function, which includes $\lfloor \frac{k}{2}\rfloor$ constant control inputs, it can accommodate at most $\lceil \frac{k}{2} \rceil$ free variables. Hence, the proof.
\end{proof}

\begin{lemma}\label{lemma:counter}
To add $m$ operands, each $b$-bit wide, the number of required $(n:k)$ counters are at most $(b(k+m)/n)$, where $k < n < b$.
\end{lemma}
\begin{proof}
This derives directly from the multi-operand carry-save adders, where, now $(n:k)$ counters are used instead of standard $(3:2)$ counters or full-adders. Note that a column of binary values at a specific bit position generates carry bits propagating up to $(k-1)$ upper bit positions, when one $(n:k)$ counter is applied. At the LSB column, total number of bits is $m$ (from that many operands), and therefore total $\lceil\frac{m(k-1)}{n}\rceil$ carry bits are generated, corresponding to $\lceil\frac{m}{n}\rceil$ counters. For the immediately upper bit position, one carry bit from each of the counter will be appended. Therefore, for that bit position, the number of counters is $\frac{m + \lceil\frac{m}{n}\rceil}{n}$, which could be upper bounded with $\frac{m + (\frac{m}{n} + 1)}{n}$. Assuming $m > k$, the MSB position requires at most ($\frac{m}{n} + (\frac{m}{n^2} + \frac{1}{n}) + \cdots + (\frac{m}{n^{k-1}} + \frac{1}{n})$) counters. Hence, the total number of counters is upper bounded by
\begin{equation*}
\begin{split}
& b \times (\frac{m}{n} + (\frac{m}{n^2} + \frac{1}{n}) + \cdots + (\frac{m}{n^{k-1}} + \frac{1}{n})) \\
& = bm \times (\frac{k-2}{mn} + \frac{1}{n} + \frac{1}{n^2} + \cdots + \frac{1}{n^{k-1}}) \\
& = bm \times (\frac{k-2}{mn} + \frac{n^{k-1} - 1}{n^{k-1}(n-1)}) \\
& \leq b(k+m)/n.
\end{split}
\end{equation*}
\end{proof}

\begin{lemma}\label{lemma:counter-maj}
    An $(n:k)$ counter can be represented using a MIG, if the constituent majority function has at least additional $(n-1)$ constant control inputs, where $n$ is odd.
    \end{lemma}
    \begin{proof}
    Let us assume the $n$ \textit{free} variables of the $(n:k)$ counter to be $\{x_1, x_2, \cdots, x_n\}$. A majority function using all the free variables is of form $M_n(x_1, x_2, \cdots, x_n)$. Clearly $M_n$ can only distinguish between two cases, if the HW of the inputs is $\geq \lceil \frac{n}{2} \rceil$ or otherwise. In order to produce the correct counter output, the constituent majority function needs to demarcate each possible HW. Considering the corner case of an $(n:k)$ counter having a HW of $0$, it is possible to introduce \mbox{$(n-1)$}~constant control variables with value set as $1$ and using $M_{2n-1}(x_1, x_2, \cdots, x_n, 1, 1, \cdots, 1)$, which gives a result of $0$. The same function will produce output of $1$ if the input HW of $\{x_1, x_2, \cdots, x_n\}$ is $1$. Progressing in the same manner, each HW can be distinguished by modifying the values of control variables.
    \end{proof}

\begin{figure}[ht]
     \centering
     \includegraphics[width=\widefig]{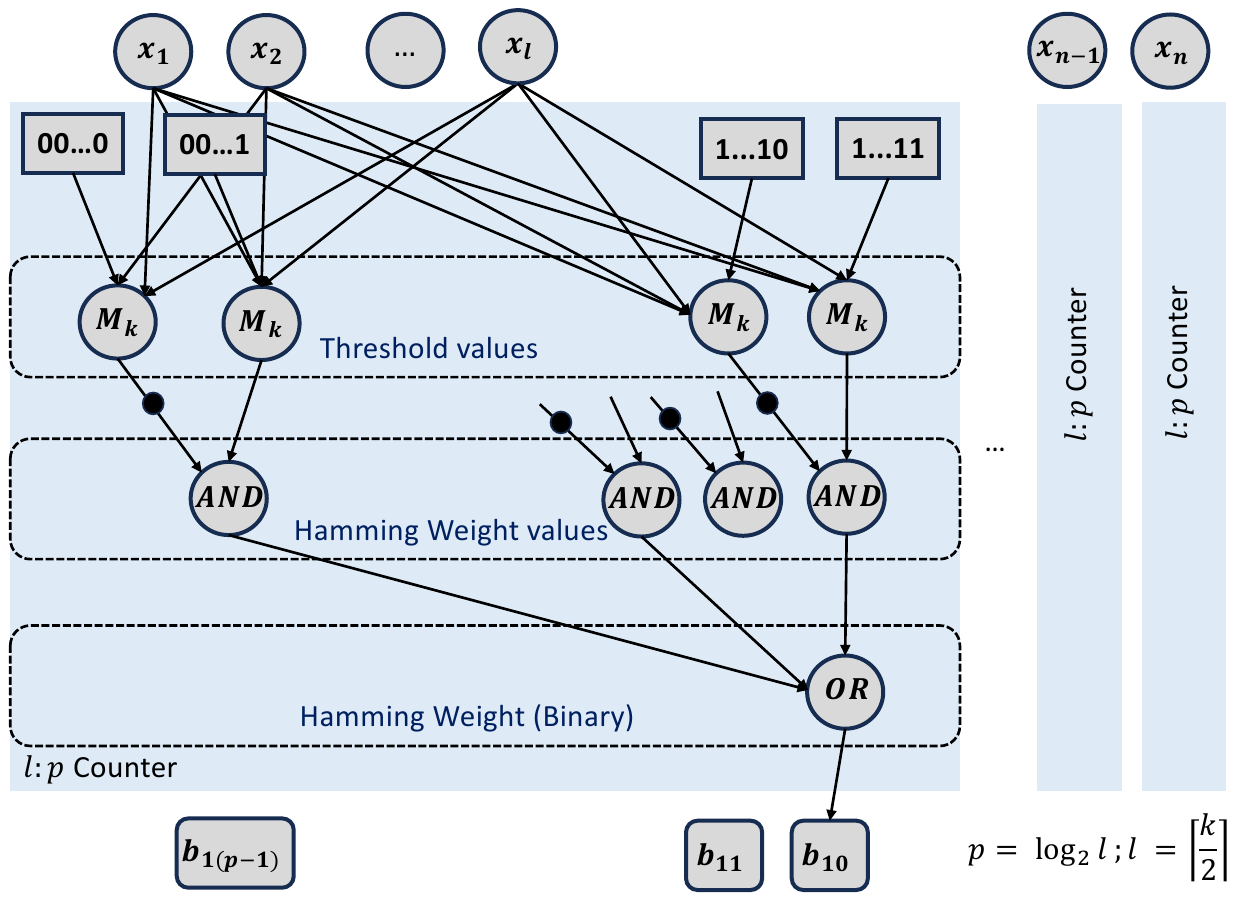}
     \caption{Generic Counter Graph Design for decomposing $M_n$ using $M_k$. AND, OR logic functions can be easily expressed using $M_k$ though by introducing redundant constants.}
     \label{fig:gen_counter}
\end{figure}

Utilizing these results, one can begin the $M_n$ decomposition by first partitioning the inputs in sets of appropriate size~($l = \ceil{\frac{k}{2}}$). Consequently, for each partition, one $(l:p)$ counter (where $p = \log_2(l)$) is required to generate the HW values in binary form, as shown in Figure~\ref{fig:gen_counter}. Inside the counter, $M_k$ functions are utilized, for which the construction could be described in three steps as following.
\begin{itemize}
    \item \textit{Stage I:} Consider all the inputs, and append the constant values to it, for determining a specific threshold value. For example, $$M_k(x_1, x_2, \cdots, x_l, 0, 0, \cdots, 0) = 1$$ indicates all inputs from $x_1$ to $x_l$ are $1$, where $l = \lceil k/2 \rceil$, following corollary~\ref{corollary:counter}. The outputs of this stage indicate $T_l^1$, $T_l^2$ and so on.
    \item \textit{Stage II:} If two consecutive threshold values are AND-ed with the upper one complemented, the resulting outcome is a HW signal. In other words, if a set of inputs is $T_l^1$ and false for $T_l^2$, we can confirm it to have HW 1.
    \item \textit{Stage III:} In this stage, we convert HW signals to a binary HW value. For that, logical OR is performed with multiple possible inputs, e.g., LSB of the HW value is true if the HW signal is originating from a location that is an odd number, e.g., $HW_1$, $HW_3$ and so on.
\end{itemize}

\begin{theorem}\label{thorem:counter_tree}
Following the counter graph decomposition approach, $M_n$ can be realized using $\mathcal{O}(n)$ such many $M_k$ functions.
\end{theorem}
\begin{proof}
We split the proof in two parts. First, we present the number of counters needed, followed by the number of $M_k$ functions needed for each counter.

The $n$ inputs are partitioned with each set holding $l$ inputs, where \mbox{$l = \lceil k/2 \rceil$}. Therefore, the number of required partitions $N_p$ is as follows:
\begin{equation}
N_p = \frac{n}{l} = \frac{n}{\lceil k/2 \rceil} \leq \frac{2n}{k}\\
\end{equation}
It may be noted that one $(l:p)$ counter is needed for each partition to generate the HW, where $p = \log_2\lceil\frac{k}{2}\rceil$. Furthermore, $(l:p)$ counters are needed to add the HW bits from each partition with the threshold value to generate the final $M_n$ bit. This is essentially a multi-operand adder with $(2n/(k + 1))$ operands, each up to $p$-bit wide except the threshold value, which can be $\log_2(n)$-bit wide. Ignoring that specific operand, we obtain the number of $(l:p)$ counters to be at most $p(2n/(k + 1) + p)/l$, following Lemma~\ref{lemma:counter}. Summing these components with the counters needed for each partition, we obtain the total number of $(l:p)$ counters as follows, taking
$l = \frac{(k+1)}{2}$, for the ceiling function.
\begin{equation*}
\begin{split}
& \frac{p(2n/(k+1) + p)}{l} + \frac{2n}{(k+1)} \\
& = \frac{2(2n + kp + p)p}{(k+1)(k+1)} + \frac{2n}{(k+1)} \\
& = \frac{2p(2n + kp + p) + 2n(k+1)}{(k+1)^2}
\end{split}
\end{equation*}

For realizing the $(l:p)$ counter using $M_k$, we follow the three stages depicted in Figure~\ref{fig:gen_counter}. Both stage one and stage two, computing the threshold values and Hamming values, respectively - requires $l$ many $M_k$ functions. The last stage requires $p$ many $M_k$ functions. Altogether, this leads to $(2l+p)$ $M_k$ functions for each $(l:p)$ counter. For simplicity of bound calculation, we put $p = \log_2k$ and ignore lower order terms to obtain the total number of $M_k$ functions needed as following.

\begin{equation}
\begin{split}
& \frac{2(p(2n + kp + p) + n(k+1))(2l + p)}{(k+1)^2} \\
\leq& \frac{2(2n\log_2k + k(\log_2k)^2 + nk)(k + \log_2k)}{(k+1)^2}\label{eq:bound}
\end{split}
\end{equation}

Considering the highest order term, we get the bound on the number of $M_k$ function required to decompose $M_n$ to be $\mathcal{O}(\frac{nk^2}{(k+1)^2})$, i.e., $\mathcal{O}(n)$.
\end{proof}

\begin{corollary}\label{cor:counter_tree_threshold}
Following the counter graph decomposition approach, an $n$-input Threshold Boolean function can be realized using $\mathcal{O}(n)$ such many $M_k$ functions.
\end{corollary}
\begin{proof}
The result directly follows from theorem~\ref{thorem:counter_tree}, where only the threshold value to be added in the intermediate counter graph composition steps is different.
\end{proof}


Note that, in practice, multiple optimizations can be applied to reduce the number of $M_k$ functions further.

\section{Decomposition Procedure: Partition Function Approach}\label{sec:partition_function}
\noindent The decomposition using partition function approach proceeds in three phases. \textit{First}, the input bits are partitioned into multiple groups of input count $l$, where $l = \lceil\frac{k}{2}\rceil$. For each group, $l$ output bits are produced indicating the HW of the input $1..l$, i.e., exactly one output bit is set to True. \textit{Second}, the same HW bits from all the groups are processed to identify if at least $t$ HW bits are true. \textit{Finally}, the restricted set partition function is invoked to combine the threshold stage outputs. This is exemplified through a decomposition of $M_9$ to $M_5$ in the following Figure~\ref{fig:partition_function_9_5}.

\begin{figure}[hbt]
     \centering
     \includegraphics[width=\linewidth]{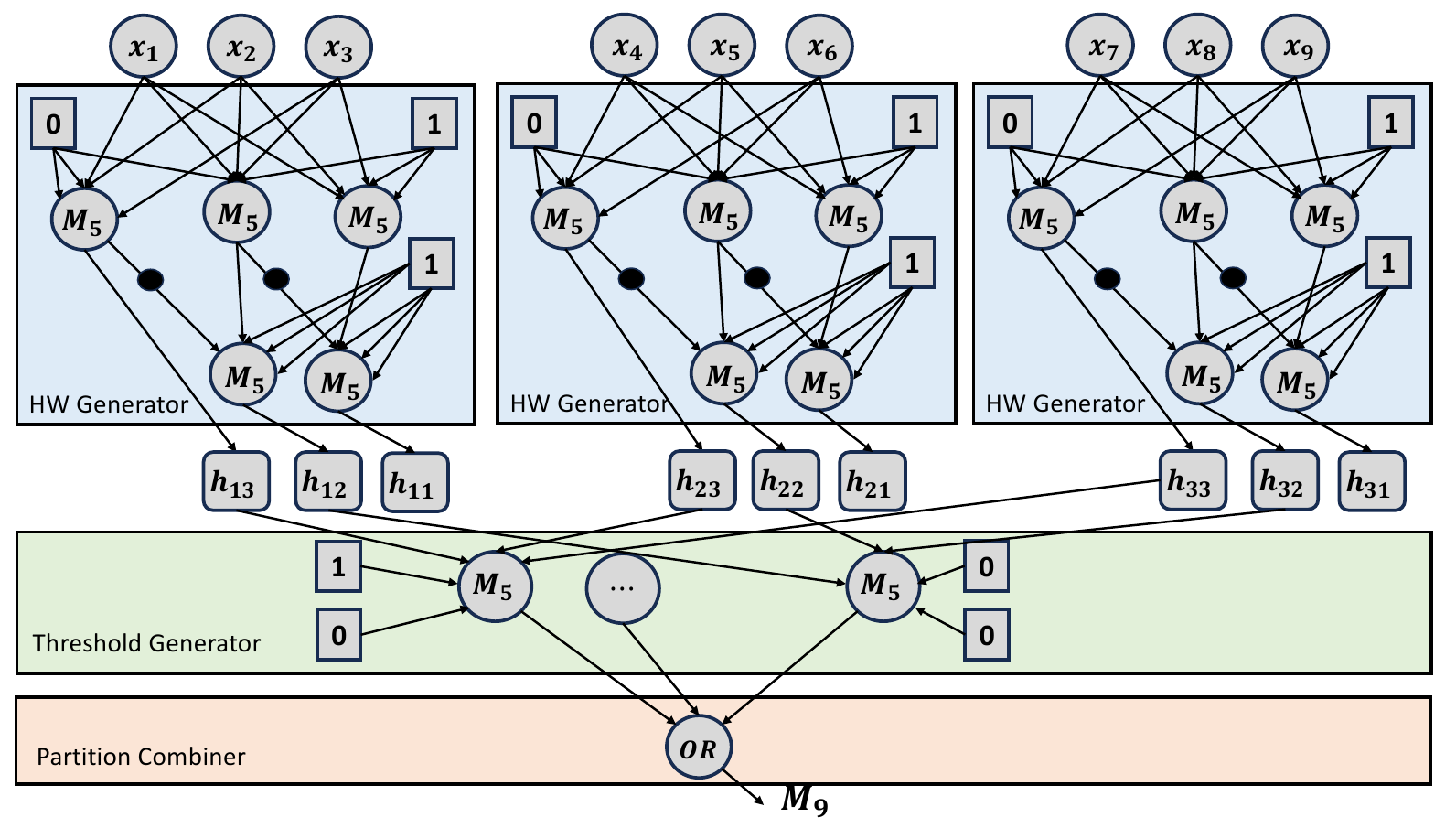}
     \caption{$M_9$ computation using $M_5$ via Partition Function Approach.}
     \label{fig:partition_function_9_5}
\end{figure}

In the final stage, we need to identify the partitions of $5$, where the maximum value of a constituent number is $3$ (corresponding to the HW of a group) and there can be at most $3$ elements in the partition (corresponding to $3$ groups). Hence, we need to determine the partitions corresponding to $p_r(3, 3, 5)$, which are, $\{3+2\}$ and $\{2+2+1\}$. However, since we are also interested in achieving any possible overall \mbox{HW $\ge 5$}, we can identify the partitions for $p(3, 3, 6)$ to obtain $\{3+3\}$, $\{2+2+2\}$ and $\{3+2+1\}$. In this case, one may note that since, $\{3+2+1\}$ (as one of the partitions of $6$) includes $\{3+2\}$ (as one of the partitions of $5$), it is redundant to consider $\{3+2+1\}$ in the last phase. Therefore, we define a new function that is larger in scope compared to restricted partition function. We denote such function as \textit{Restricted Set Partition Function}, indicating it covers a set of numbers to partitioned instead of a single number. Formally, it is defined as follows.

\begin{definition}
Restricted Set Partition Function $p_{rs}(N, M, n)$ is the total number of partitions of all numbers $z$ (where $n \le z \le N\times M$) using at most $M$ parts, where each part is at most $N$.
\end{definition}

In what follows, we attempt to obtain the bound of compositional complexity using the partition function approach. Let us first derive some preliminary results necessary for that.

\begin{lemma}\label{lemma:rsp-bound}
Considering $N\times M = (2n-1)$, Restricted Set Partition Function $p_{rs}(N, M, n)$ is upper bounded by $\frac{1}{8\sqrt{3}}e^{\pi\sqrt{\frac{4n}{3}}}$.
\end{lemma}
\begin{proof}
Note that, we apply the constraint of $N\times M = (2n-1)$ as it applies in our case. Taking the partition function without any restriction, the maximum number could be $n$ times the bound outlined in equation~\eqref{eqn:ramanujan}. Hence, we obtain,
\begin{equation*}
\begin{split}
& n\frac{1}{4(2n-1))\sqrt{3}}e^{\pi\sqrt{\frac{2(2n-1))}{3}}} \\
& \le \frac{1}{8\sqrt{3}}e^{\pi\sqrt{\frac{4n}{3}}} \\
\end{split}
\end{equation*}

\end{proof}

When trying to construct $n$-input OR gate using $M_k$, it suffices to have $1$ $M_k$ gate, if $n \leq \frac{k}{2}$. Otherwise, we have the following result.

\begin{lemma}\label{lemma:or-gate}
To compute an $n$-input OR gate, using $M_k$ (where $n > \frac{k}{2}$), the number of required $M_k$ gates is given by $\frac{2(n-1)}{k}$.
\end{lemma}
\begin{proof}
The largest OR gate that can be computed using one $M_k$ contains $l$ inputs, where $l = \lceil\frac{k}{2}\rceil$. One can construct a logarithmic depth tree for computing $n$-input OR using $l$-input OR gates, where, the depth of the tree is $\lceil \log_ln \rceil$. The total number of $l$-input OR gates in that structure is given by $$\frac{n-1}{l-1} = \frac{(n-1)}{\lceil\frac{k}{2}\rceil - 1} = \frac{(n-1)}{\frac{k}{2}} = \frac{2(n-1)}{k}$$
\end{proof}

\begin{corollary}\label{corollary:and-gate}
To compute an $n$-input AND gate, using $M_k$ (where $n > \frac{k}{2}$), the number of required $M_k$ gates is given by $\frac{2(n-1)}{k}$.
\end{corollary}

\begin{figure}[hbt]
     \centering
     \includegraphics[width=0.8\linewidth]{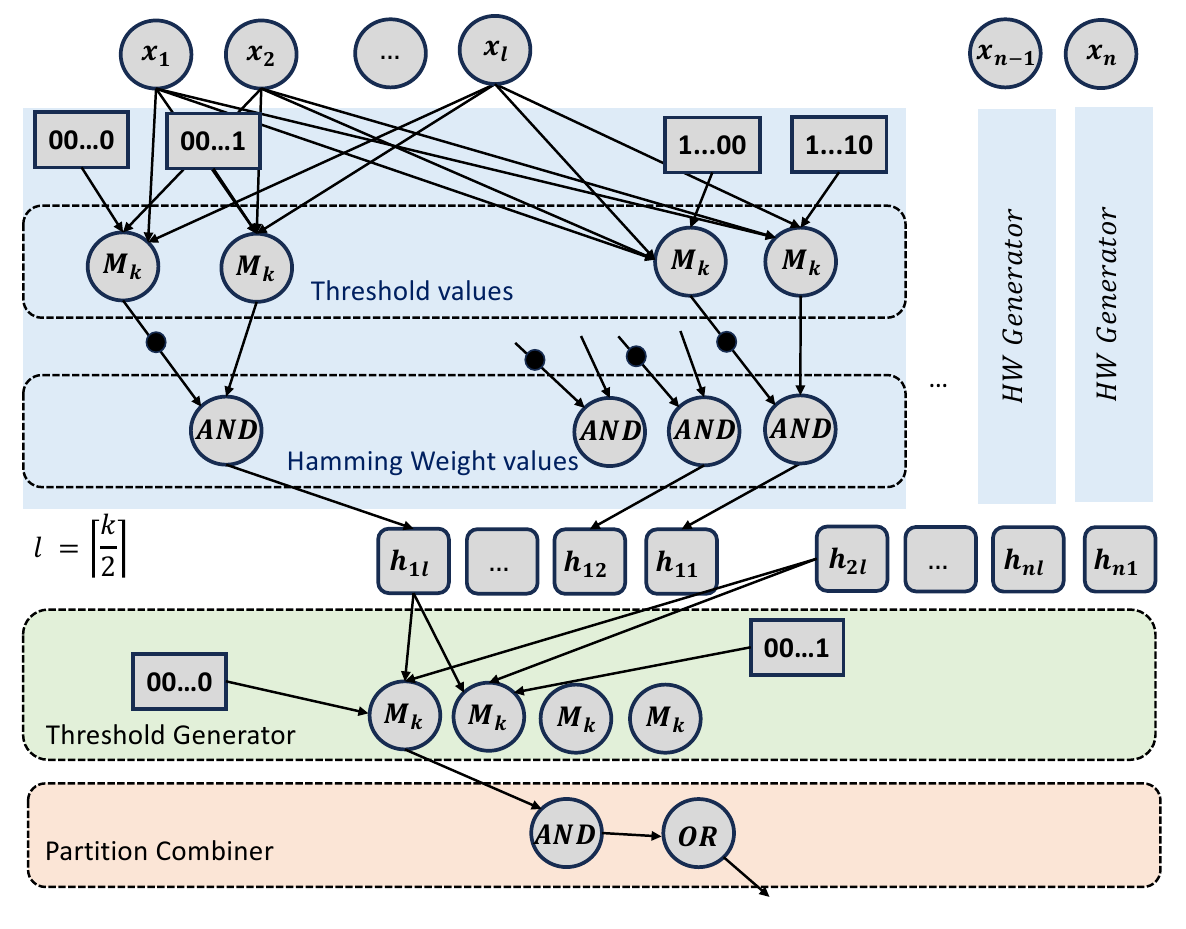}
     \caption{$M_n$ computation using $M_k$ via Partition Function Approach.}
     \label{fig:partition_function_n_k}
\end{figure}

The generalized decomposition flow using partition function approach is depicted in the Figure~\ref{fig:partition_function_n_k}.

\begin{theorem}\label{thorem:partition_function}
Following the partition function decomposition approach, $M_n$ can be realized using $\mathcal{O}(\frac{n}{k^2}e^{\sqrt{n}})$ $M_k$ functions.
\end{theorem}
\begin{proof}
Similar to the prior decomposition, the $n$ inputs are partitioned with each set holding $l$ inputs, where $l = \lceil k/2 \rceil$. Hence, the number of required partitions $N_p \leq \frac{2n}{k}$. In the first two stages of the decomposition, the HW values are calculated. For that, in each partition, $2l$ $M_k$ functions are needed.

In the next phase, the HW values from individual sets are combined to obtain corresponding threshold values, i.e., one needs to determine if there are at least $t$ HW of $1, 2, \cdots, \frac{2n}{k}$, for values of $t$ from $1$ to $l$. Considering $\frac{2n}{k} > k$, here, we can invoke Corollary~\ref{cor:counter_tree_threshold} to obtain total number of $M_k$ functions to be of the following order.

\begin{equation*}
\begin{split}
& l \times \frac{n}{l} \times \mathcal{O}(\frac{2n}{k}) = n \times \mathcal{O}(\frac{2n}{k}) = \mathcal{O}(\frac{2n^2}{k})
\end{split}
\end{equation*}

The threshold outputs are then AND-ed at the last phase of partition combiner. For a single \texttt{and} gate, at most $\frac{2n}{k}$ inputs are needed, one from each set. Following Corollary~\ref{corollary:and-gate}, every \texttt{and} gate requires $\frac{2(2n/k-1)}{k}$ $M_k$ gates. The total number of such \texttt{and} gates, as well as the number of inputs to the final \texttt{or} gate is determined by the restricted set partition function, which is upper bounded by $\mathcal{O}(e^{\sqrt{n}})$. Since this dictates the number of required \texttt{and} gates, the highest order term, across all phases of the decomposition is given as $\mathcal{O}(\frac{n}{k^2}e^{\sqrt{n}})$.
\end{proof}

It can be noted that the decomposition using partition function invokes the decomposition using counter graph approach within it. This could be avoided though the final bound order remains the same. This result could also be accompanied with a corresponding corollary for threshold function realization using partition function approach. Since, this decomposition turns out to be much less efficient compared to the counter graph decomposition, we refrained from detailed experimental evaluation of this.

\input{results}

\noindent

\section{Concluding Remarks}\label{sec:conclusion}
\noindent Decomposition of large-input Majority Boolean functions have been studied for at least six decades for its relevance in circuit complexity theory and more recently for taking advantage of emerging computing devices. In this work, for the first time, we introduce a construction that can express $M_n$ using $\mathcal{O}(n)$ many such $M_k$ functions. The decomposition, using a counter-tree approach, is closest to the optimal lower bound reported recently~\cite{complexity_ccc}. We also explored an alternative construction using partition function, where the theoretical complexity turns out to be worse. The counter graph-based construction is experimentally validated demonstrating excellent match between experimental results and theoretical bounds. We hope that this study will influence further research in connecting theoretical results to practical implementations. The possibility to utilize counter graphs for general logic structure manipulation in Majority and Threshold logic systems remain to be explored as well.

\section*{Acknowledgment}
\noindent Subhamoy Maitra acknowledges the funding support provided by the ``Information Security Education and Awareness (ISEA) Project phase - III, Cluster - Cryptography, initiatives of MeitY, Grant No. F.No. L-14017/1/2022-HRD".



%% file: results.tex
\section{Experimental Studies}\label{sec:exp}
\noindent We will explain the implementation flow and the experimental results in this section.

\begin{figure}[hbt]
  \begin{subfigure}[t]{.5\columnwidth}
    \centering
    \caption{}
    \label{fig:mlir_flow}
    \includegraphics[width=3.5cm]{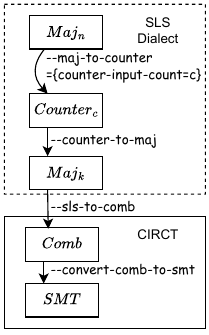}
  \end{subfigure}
  \begin{subfigure}[t]{.4\columnwidth}
  \caption{}
    \label{fig:sls_dialect}
    {\scriptsize
    \begin{tabular}{rcc}
    \bottomrule
    \textbf{Operation} & \textbf{\#Inputs} & \textbf{\#Outputs} \\
    \hline
    \texttt{not} & single & single \\
    \hline
    \texttt{and} & multi & single \\
    \hline
    \texttt{or} & multi & single \\
    \hline

    \texttt{maj} & multi (odd) & single \\
    \hline
    \texttt{counter} & multi & multi \\
    \toprule
    \end{tabular}
    }
  \end{subfigure}

  \caption{(\subref{fig:mlir_flow}) MLIR-based flow for Majority Decomposition. The SLS dialect has been implemented and integrated with CIRCT, along with the passes {\footnotesize\{\texttt{--maj-to-counter, --counter-to-maj, --sls-to-comb}\}}. (\subref{fig:sls_dialect}) The operations defined in SLS dialect.}
\end{figure}

\subsection{Implementation Flow}\label{flow}
\noindent In this section, we present the evaluation of our Boolean decomposition approach implemented using custom operators and transformation passes in CIRCT~\cite{circt}, an MLIR~\cite{lattner2021mlir}-based framework for hardware compilation.
To clarify the terminology for readers unfamiliar with MLIR, an {\em operator} in this context refers to a fundamental operation or computation defined within the MLIR framework, which can be customized to represent specific functionalities, such as Boolean logic operations. A {\em transformation pass} is a modular component in MLIR that applies optimizations or transformations to the intermediate representation~(IR) of a program. These passes enable the restructuring or simplification of the IR, such as decomposing complex Boolean expressions into simpler, more efficient forms. Together, custom operators and transformation passes form the backbone of our implementation, allowing us to achieve the desired Boolean decomposition efficiently within the CIRCT ecosystem.

The overall implementation flow is shown in Figure~\ref{fig:mlir_flow}.
We propose a new dialect, namely Structured Logic Synthesis~(SLS) dialect in MLIR, that defines the basic Boolean operators for logic synthesis in the context of this paper, namely \texttt{not}, \texttt{and}, \texttt{or},
$\texttt{counter}$ and majority~(\texttt{maj}). \texttt{not} has a single input operand, \texttt{maj} has odd number of inputs, whereas the rest of the operators support two or more inputs.  \texttt{counter} has multiple outputs, while the rest of them have a single output. The operations defined in the dialect is summarily listed in Figure~\ref{fig:sls_dialect}.

Using the proposed flow, the input $n$-input majority is lowered (\texttt{--maj-to\newline-counter}) into $c$-input counter operators, where $k= 2c - 1$. In the following pass, each counter is lowered~(\texttt{--counter-to-maj}) into a combination of operators, as shown in the example of Figure~\ref{fig:gen_counter}. Fig~\ref{fig:decomposition} shows a detailed step-by-step demonstration of decomposition of $Maj_{11}$ into $Maj_{9}$ using the proposed flow.
All experiments were conducted on a machine equipped with an Apple M1 Ultra system-on-chip (SoC), featuring a total of 20 CPU cores. The system is configured with 128 GB of unified memory, ensuring sufficient capacity for computationally intensive tasks. The operating system used is macOS~\texttt{15.2}~(build 24C101), running on a Darwin kernel version~\texttt{24.2.0}. The machine architecture is \texttt{arm64}. The implementation was done on top of the  CIRCT version~\texttt{da2ca8c}. All the decompositions took between ${0.046s-0.554s}$
with a standard deviation of $0.0866$. Furthermore, the output of decomposition was verified formally using the \texttt{circt-lec} tool~\cite{circtformal}, that uses the z3 SMT library~\cite{Z3Prover}.
The number of majority nodes in the final decomposition can be improved further using algorithmic techniques such as output-based pruning, majority rewriting and other techniques proposed in literature~\cite{mathias_xor_maj,decomp_tcad}.

\subsection{Results}
\begin{figure}[ht]
    \centering
    \includegraphics[width=\widefig]{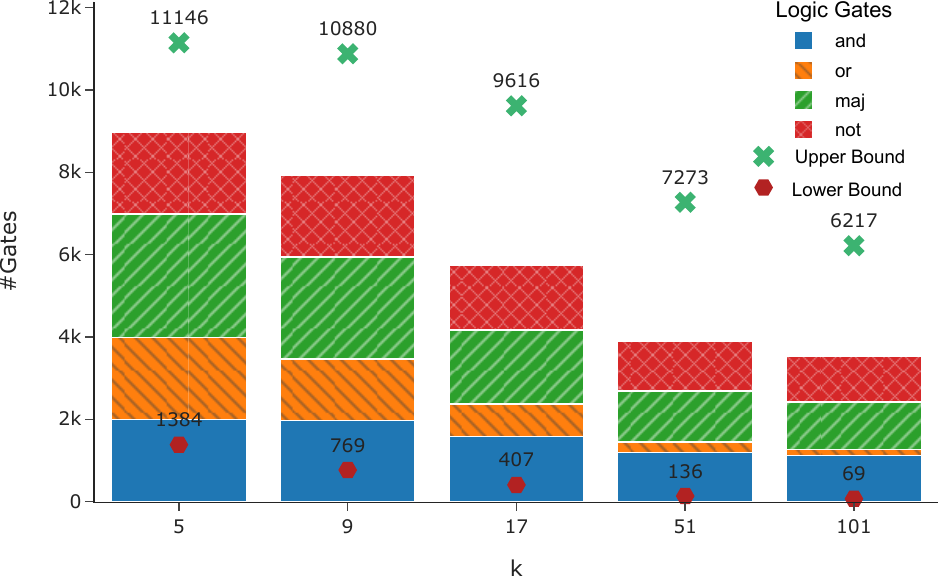}
    \caption{Decomposition of $M_{1001}$ into various $M_k$. As per the counter graph decomposition procedure, the \texttt{or} and \texttt{and} gates obtained during decomposition,  can be expressed in terms of majority~(\texttt{maj}) using the Lemma~\ref{lemma:or-gate} and Corollary~\ref{corollary:and-gate} respectively. The theoretical upper bound~\eqref{eq:bound} is marked with green~$\times$ and the optimal lower bound $\Omega(\frac{n}{k}\log k)$ is marked with a red dot. }
    \label{fig:majnk_allk}
\end{figure}

We present analysis of decomposition of Majority-1001~($n=1001$) into a variety of $M_k$ targets in Figure~\ref{fig:majnk_allk}. With the increase in $k$, the total number of gates reduce in the decomposition, as expected. Similarly, Figure~\ref{fig:majnk_alln} shows the impact of decomposition of various $M_n$ into $M_9$($k=9$) gates. As can be observed in both the figures, the decomposition procedure involves \texttt{and}, \texttt{or} and \texttt{not} gates. Considering the decomposition to be realized in terms of MIG, the \texttt{or} and \texttt{and} gates could be expressed as \texttt{majority} function  using the Lemma~\ref{lemma:or-gate} and Corollary~\ref{corollary:and-gate} respectively, resulting in a homogeneous decomposition.

\begin{figure}[ht]
    \centering
    \includegraphics[width=\widefig]{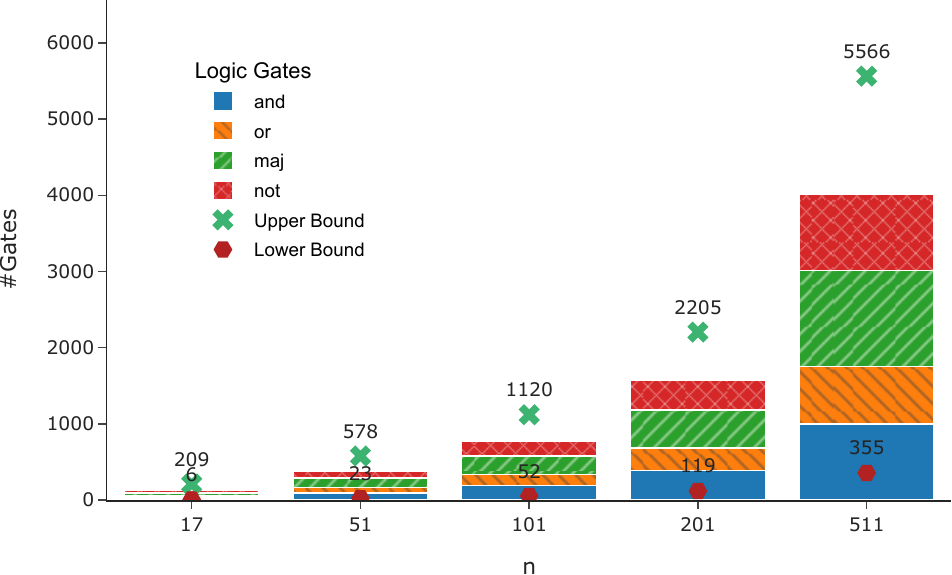}
\caption{Decomposition of $M_{n}$ into  $M_9$. The theoretical upper bound Eq~\eqref{eq:bound} is marked with green~$\times$ and the optimal lower bound $\Omega(\frac{n}{k}\log k)$ is marked with a red dot.}
    \label{fig:majnk_alln}
\end{figure}

Figure~\ref{fig:majnk} presents the total number of gates required in decomposition, for all odd values of $n$ between 5 and 511 and various values of $k$. Do note that when $n < k$, a single majority gate~$M_k$ is used with constants~($\frac{n-k}{2}$ inputs are set to 1 and $\frac{n-k}{2}$ are set to 0) to realize $M_n$.

\begin{figure}[ht]
    \centering
    \includegraphics[width=\widefig]{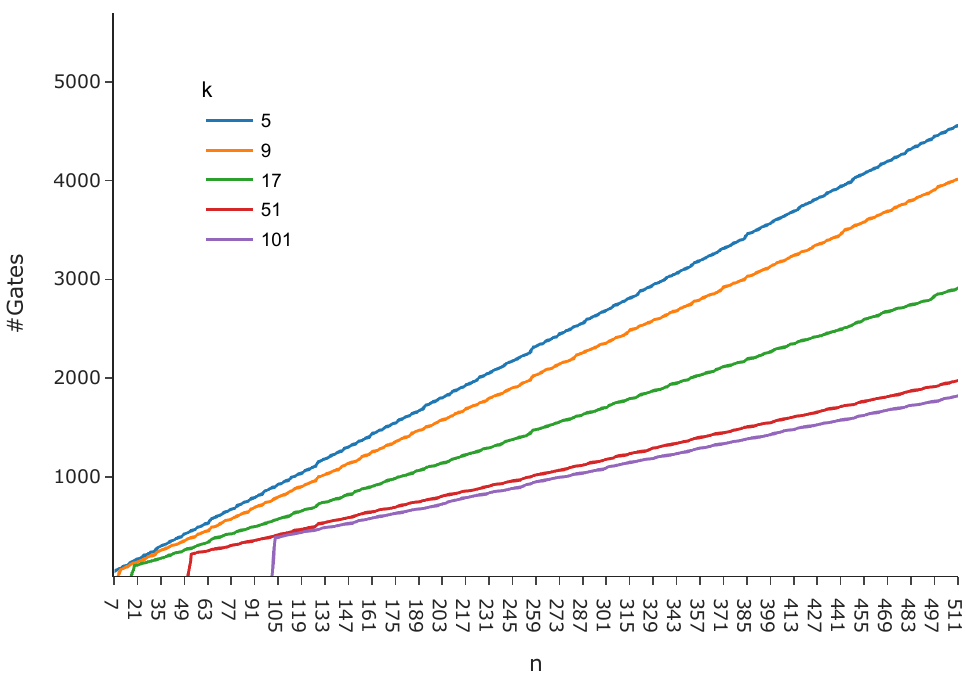}
    \caption{Decomposition of $M_{n}$ into various $M_k$, where \mbox{($5 \le n \le 511)$} and \mbox{$k=\{5,9,17,51,101\}$}.}
    \label{fig:majnk}
\end{figure}



\begin{figure}[hbt]
    \centering
\begin{subfigure}[t]{\columnwidth}
\centering
\caption{$Maj_{11}$ expressed in SLS dialect, in a hardware module of CIRCT.}
    \label{fig:maj11}
    {\scriptsize

    \begin{verbatim}
module {
    hw.module  @test(in %arg0: i1, in %arg1: i1,..., in %arg10: i1, out o1: i1) {
        %o1 = sls.maj(%arg0, %arg1,...,%arg10:
            i1, i1, i1, i1, i1, i1, i1, i1, i1, i1, i1)
            to i1
        hw.output %o1: i1
    }
}
    \end{verbatim}
    }
\end{subfigure}
\begin{subfigure}[t]{\columnwidth}
\centering
\caption{$Maj_{11}$ decomposed in terms of counters with $c=5$ inputs SLS dialect, using \texttt{--maj-to-counter} pass.}
    \label{fig:maj11-decomp-counter}
{\scriptsize
\begin{verbatim}

module {
  hw.module @test(in %arg0: i1, in %arg1: i1,..., in %arg10: i1, out o1: i1) {
    %0 = sls.constant {value = false}
    %1 = sls.constant {value = true}
    %2:3 = sls.counter(%arg0, %arg1, %arg2, %arg3, %arg4 : i1, i1, i1, i1, i1) to i1, i1, i1
    %3:3 = sls.counter(%arg5, %arg6, %arg7, %arg8, %arg9 : i1, i1, i1, i1, i1) to i1, i1, i1
    %4:3 = sls.counter(%arg10, %0, %2#0, %3#0 : i1, i1, i1, i1) to i1, i1, i1
    %5:3 = sls.counter(%1, %2#1, %3#1, %4#1 : i1, i1, i1, i1) to i1, i1, i1
    %6:3 = sls.counter(%0, %2#2, %3#2, %4#2, %5#1 : i1, i1, i1, i1, i1) to i1, i1, i1
    %7:2 = sls.counter(%1, %5#2, %6#1 : i1, i1, i1) to i1, i1
    %8:2 = sls.counter(%6#2, %7#1 : i1, i1) to i1, i1
    hw.output %8#0 : i1
  }
}

% \end{minted}
\end{verbatim}
}

\end{subfigure}
\begin{subfigure}[t]{\columnwidth}
\centering
\caption{The counters with $c=5$ composed using $Maj_{9}$, $and$, $or$ and $not$ operations, using \texttt{--counter-to-maj} pass.}
    \label{fig:maj11}
    {\scriptsize

          \begin{verbatim}

module {
  hw.module @test(in %arg0: i1, in %arg1: i1,..., in %arg10: i1, out o1: i1) {
    %0 = sls.constant {value = false}
    %1 = sls.constant {value = true}
    %2 = sls.maj(%arg0, %arg1, %arg2, %arg3, %arg4, %1, %1, %1, %1 : i1, i1, i1, i1, i1, i1, i1, i1, i1) to i1
    %3 = sls.maj(%arg0, %arg1, %arg2, %arg3, %arg4, %0, %1, %1, %1 : i1, i1, i1, i1, i1, i1, i1, i1, i1) to i1
    %4 = sls.maj(%arg0, %arg1, %arg2, %arg3, %arg4, %0, %0, %1, %1 : i1, i1, i1, i1, i1, i1, i1, i1, i1) to i1
    ...
    hw.output %73 : i1
  }
}
          \end{verbatim}
    }
\end{subfigure}
    \caption{Expressing the $Maj_{11}$~($n=11$) using $M_9$~($k=9)$ operations using the proposed flow. $argi$ indicates the $i^{th}$ input and $o1$ is the output. $i1$ indicates Boolean inputs and output.}
    \label{fig:decomposition}
\end{figure}